\documentclass[a4paper,11pt]{article}

\usepackage[utf8]{inputenc}	
\usepackage[T1]{fontenc}
\usepackage[english]{babel}
\usepackage{amssymb}
\usepackage{amsmath}
\usepackage[bookmarks]{hyperref}
\usepackage{graphicx}
\usepackage{booktabs}
\usepackage{todonotes}
\usepackage{enumitem}
\usepackage{amsthm}
\usepackage{multirow}
\usepackage{xspace}
\usepackage{url}

\usepackage{geometry}
\newcommand{\destBrib}{\textsc{DB}\xspace}
\newcommand{\negative}{\textsc{negative}\xspace}
\newcommand{\weigh}{\textsc{weighted}\xspace}
\newcommand{\Knapsack}{\textsc{Knapsack}\xspace}
\newcommand{\NP}[0]{$\mathcal{NP}$\xspace}
\newcommand{\Pclass}[0]{$\mathcal{P}$\xspace}
\newcommand{\NAEdreiSAT}{\textsc{NAE-3SAT}\xspace}
\newcommand{\DreiSAT}{\textsc{3SAT}\xspace}
\newcommand{\budget}{\beta\xspace}
\newcommand{\Issue}{X\xspace}
\newcommand{\issue}{x\xspace}
\newcommand{\etal}{\textit{et al.}\xspace}
\newcommand{\NPcResult}[1]{\hyperref[#1]{$\mathcal{NP}$\textbf{-c}}}
\newcommand{\PResult}[1]{\hyperref[#1]{$\mathcal{P}$}}
\newcommand{\bigO}[1]{\mathcal{O}(#1)\xspace}
\newcommand{\IV}{IV\xspace}
\newcommand{\DV}{DV\xspace}
\newcommand{\IDV}{IV+DV\xspace}
\newcommand{\OP}{OP\xspace}
\newcommand{\OK}{OK\xspace}
\newcommand{\OKeff}{OK^\textit{eff}\xspace}
\newcommand{\OV}{OV\xspace}
\newcommand{\SM}{SM\xspace}
\newcommand{\Cany}{C_\textsc{any}\xspace}
\newcommand{\Cequal}{C_\textsc{equal}\xspace}
\newcommand{\Clevel}{C_\textsc{level}\xspace}
\newcommand{\Cflip}{C_\textsc{flip}\xspace}
\newcommand{\Cdist}{C_\textsc{dist}\xspace}
\newcommand{\satc}{\gamma\xspace}
\newcommand{\satm}{\mathfrak{m}\xspace}
\newcommand{\satv}{\nu\xspace}
\newcommand{\satn}{\mathfrak{n}\xspace}

\newtheorem{theorem}{Theorem}

\newtheorem{corollary}[theorem]{Corollary}

\usepackage{authblk}

\begin{document}

\title{Often harder than in the Constructive Case: Destructive Bribery in CP-nets}

\author[1]{Britta Dorn}
\author[2]{Dominikus Krüger}
\author[2]{Patrick Scharpfenecker\thanks{Supported by DFG grant TO 200/3-1.}}

\affil[1]{Faculty of Science, Dept.~of Computer Science,
			University of Tübingen,
			Germany, 
			\texttt{britta.dorn@uni-tuebingen.de}}

\affil[2]{Institute of Theoretical Computer Science,
			Ulm University,
			Germany,
			\texttt{\{dominikus.krueger,patrick.scharpfenecker\}@uni-ulm.de}}

\date{}

\maketitle

\begin{abstract}
We study the complexity of the destructive bribery problem---an external agent tries to prevent a disliked candidate from winning by bribery actions---in voting over combinatorial domains, where the set of candidates is the Cartesian product of several issues. This problem is related to the concept of the margin of victory of an election which constitutes a measure of robustness of the election outcome and plays an important role in the context of electronic voting. In our setting, voters have conditional preferences over assignments to these issues, modelled by CP-nets. 
We settle the complexity of all combinations of this problem based on distinctions of four voting rules, five cost schemes, three bribery actions, weighted and unweighted voters, as well as the negative and the non-negative scenario. We show that almost all of these cases are \NP-complete or \NP-hard for weighted votes while approximately half of the cases can be solved in polynomial time for unweighted votes. 
\end{abstract}

\section{Introduction}
Voting in an election is a common procedure to aggregate the preferences of the parties involved, the voters, over a set of alternatives, the
candidates, in order to find one or more winning alternatives. In many settings, the set of candidates is the Cartesian product of several issues. One might think of a referendum, where voters have to approve or disapprove of each issue, or the individual configuration of a product consisting of several components for each of which several options can be chosen, such as a car where the consumer can choose between different options for the model,
equipment, color, and various other features, or a computer where
different options are available for the operation system, hardware and
software components.
The number of possible candidates (available choices, outcomes) is
hence exponential in the number of issues or components, and it
may be an impossible task for voters to express their preferences over
the whole set of available choices by ranking them all. 

Additionally, voters might have conditional preferences over the
candidates. The typical example is a meal consisting of several components, such as a main dish (fish or meat), a side dish (chips or rice), and a drink (beer or wine). A voter might (unconditionally) prefer meat to fish, and he might prefer wine to beer given that fish is served for the main dish. In the car example, a consumer might prefer a station wagon to a hatchback, and he might prefer a black car to a white one, but only if it is equipped with an air conditioning system.

In view of applications such as e-commerce and other settings on the
web and internet where one has to deal with very large populations,
one is interested in a compact description and efficient communication
and aggregation of these conditional preferences in combinatorial
domains. One approach is given by CP-nets,
a graphical model introduced by Boutilier~\textit{et al.}~\cite{boutilier2004CPnets} that incorporates {\it ceteris paribus} (cp) statements describing the conditional dependencies. Preference aggregation in CP-nets was studied by Rossi~\textit{et al.}~\cite{rossi2004mcp} and various other authors (e.g.,~\cite{purrington2007making,xia2008voting,conitzer2011hypercubewise}).

Besides the problem of determining a winner, a central topic in the computational social choice literature is the study of the computational complexity of voting problems such as strategic voting (manipulation), election control and bribery (\cite{brandt2013comsoc,faliszewski2009richer}). In the bribery problem, initially introduced by Faliszewski~\textit{et al.}~\cite{faliszewski2009hard} (see also~\cite{faliszewski2008nonuniform,faliszewski2009llull,elkind2009swap}), voters can be bribed to change their preferences. In the {\it constructive} bribery problem, one asks whether a briber can make his favorite candidate win the election with these changes, subject to a budget constraint. 

Mattei~{\it et al.}~\cite{mattei2013bribery} considered several procedures for determining a winner in vo\-ting with CP-nets and investigated the bribery problem in this context. They introduced and adapted several natural cost schemes for the bribery problem in the setting of CP-nets and determined the computational complexity of the problem under the various voting rules and cost schemes, also considering the level of dependency the briber can affect with his changes. In most of these cases, they obtained that the bribery problem is solvable in polynomial time. Dorn and Kr\"uger~\cite{dorn2014hardness} answered open cases and considered the weighted and negative versions of the problem. 
Further investigations of bribery in CP-nets deal with interaction and influence among voters~\cite{maran2013framework} and with representation of the voters' preferences via soft constraints~\cite{pini2013bribery}.

In this work, we study the complexity of the \textit{destructive} bribery problem in CP-nets, which asks whether a disliked candidate can be prevented from winning the election by bribery actions. 
The study of destructive bribery is also related to the concept of the {\it margin of victory} (\cite{magrino2011margin,xia2012margin,reisch2014margin}) of an election. Given a voting rule and a set of votes, the margin of victory is the minimum number of votes that must be modified in order to change the winner(s) of the election. If the voting rule selects a unique winner, then the problem of deciding whether this number is larger than a given threshold corresponds to the destructive bribery problem introduced by Faliszewski {\it et al.}~\cite{faliszewski2009hard}. The margin of victory is a measure of robustness of the outcome of an election, specifying the number of errors that may occur in an election---be it due to inference or due to fraud---without having an effect on the outcome. It is of particular interest in the setting of electronic voting where post-election audits are executed to verify the correctness of the electronical record (\cite{norden2007postelectionaudits}). An audit samples ballots and measures the discrepancy of the sampled electronic votes with respect to their paper record. Risk-limiting post-election audits seek to minimize the size of the audit when the outcome is correct (\cite{stark2009risk}). The margin of victory is an important parameter used to determine the size of an audit for this method.

We study all combinations of voting rules, cost schemes, and bribery actions considered by Mattei~{\it et al.}~\cite{mattei2013bribery}, as well as weighted voters and the negative scenario. 
 The destructive variant has been investigated in various voting problems~\cite{conitzer2003many,conitzer2002complexity,hemaspaandra2007anyone}, including bribery~\cite{faliszewski2009llull} without combinatorial domains. In all these settings, for the unique-winner case, the destructive version is at most as hard as the constructive one. 
We think that our work might be interesting for several reasons: First, in our setting, destructive bribery turns out to be harder than constructive bribery in many cases. Second, the problems we use for our reductions (two variants of the {\sc Satisfiability} problem and the \Knapsack problem) are not the typical ones that are often used in the context of voting problems. 
 An overview of our results is given in Table~\ref{table:results} on page~\pageref{table:results}.

\section{Preliminaries}

Almost all our notations and definitions can be found in greater detail and exemplified in the articles by Mattei~\etal~\cite{mattei2013bribery} and by Dorn and Krüger~\cite{dorn2014hardness}, who analyzed the constructive case of the same scenario.

This section is structured as follows. First, we present the \NP-complete problems we use for our reductions. Afterwards we define CP-nets and introduce related notation. This is followed by the introduction of the voting rules we will work with. We are then ready to define the bribery problem in the  setting of CP-nets and introduce the different cost schemes and allowed bribery actions. Finally, we give an overview (Table~\ref{table:results}) of the results obtained in this paper and close this section with an example.

\noindent For our reductions we use the following \NP-complete problems.

\begin{quote}
\textsc{(Not-All-Equal) 3-Satisfiability}, (\textsc{NAE}-)\DreiSAT \cite{GareyJohnson}\newline
\textbf{Given:} A set $U$ of $\satn$ variables $\satv_i$, collection $C$ of $\satm$ clauses over $U$ such that each clause $\satc\in C$ is a subset of $U$ with $\lvert \satc\rvert = 3$.\newline
\textbf{Question:} Is there a truth assignment for $U$ such that each clause in $C$ has at least one true (and one false) literal?
\end{quote}

\begin{quote}
\Knapsack~\cite{GareyJohnson}\newline
\textbf{Given:} A set $U$ of $n$ objects $(w_i, v_i) \in \mathbb{N}^2$ of weight $w_i\in\mathbb{N}$ and value $v_i\in\mathbb{N}$, positive integers $k,b\in\mathbb{N}$.\newline
\textbf{Question:} Is there a subset $U'\subseteq U$ of objects with total weight at most $b$ and total value at least $k$?
\end{quote}

\paragraph*{CP-nets}
In our setting, we are given a set of $m$ \textit{issues}~$M=\{\Issue_1, \dots, \Issue_m\}$, and each $\Issue_i\in M$ has a binary \textit{domain} $D(\Issue_i) = \{\issue_i, \overline{\issue_i}\}$. A complete assignment to all issues is called a {\it candidate}, so there are $2^m$ different candidates. Each of the $n$~voters has (possibly) conditional preferences over the values assigned to the issues; if the preference of an issue~$X$ depends on one or more other issues (called the \textit{parents} $Pa(X)$), we call this issue {\it dependent}, and {\it independent} otherwise. Formally, a CP-net is defined by a directed graph (with the issues as its vertices and directed edges going from $Pa(X)$ to $X$) modeling these dependencies, and a table for each issue, containing the preference over the assignment to this issue for each different complete assignment to its parents; each combination of an assignment to the parents and the corresponding preference over the issue is called a \textit{cp-statement}. For example, for a CP-net with issues $X$ and $Y$, the cp-statement $x > \overline{x}$ means that the assignment $X=x$ is unconditionally preferred to $X=\overline{x}$, while the statements $x: \overline{y} > y$ and $\overline{x}: y > \overline{y}$ express that the assignment $Y=\overline{y}$ is only preferred to $Y=y$ in the case that $X=x$ (hence, $Pa(Y) = \{X\}$ here). 
The collection of CP-nets of all voters is called a \emph{profile}.

CP-nets only define a partial order over the candidates, i.e., some candidates are incomparable. One way to expand this to strict total orders over the candidates is to give a strict total order over all issues such that no issue depends on any issue following it in this order. If the same order~$\mathcal{O}$ works out for all CP-nets of a profile, the profile is called $\mathcal{O}$-\emph{legal}~\cite{lang2007vote}.

Throughout this work, we assume that the voters' preferences on the set of issues are given by \textit{compact} (the number of parents of each issue is bounded by a constant) and \textit{acyclic} (the corresponding graph is acyclic) CP-nets. For acyclic CP-nets, the most preferred candidate of a voter can be determined efficiently~\cite{boutilier2004CPnets}. 

An example of a profile consisting of three CP-nets is given in Table~\ref{table:example} at the end of this section. The CP-nets encode conditional
preferences for the alternative options of a menu consisting of a main
dish, a side dish and a drink. Alice's choice for the drink is
dependent of the choice for the main dish: She prefers beer to wine in
case meat is served, and wine to beer if fish is served as a main
dish.

\paragraph*{Voting} 
A voting rule maps a profile to a set of candidates. 
With {\bf One-Step-$k$-Approval ($OK$)}, only the~$k$ most preferred candidates of each voter obtain $1$~point each. The winner of the election is the candidate with the most points (or all candidates with those points). In particular, we consider the special cases {\bf One-Step-Plurality ($OP$)}, where $k=1$, and {\bf One-Step-Veto ($OV$)}, where $k=2^m-1$. 
With {\bf Sequential majority ($SM$)}, 
given a total order $\mathcal{O}$ for which the profile is $\mathcal{O}$-legal, we follow this order issue by issue, and execute a majority vote for each issue. The voters fix the winning value of the corresponding issue in their CP-net and then go on to the next issue. The winning candidate is the combination of the winners of the individual steps taken. These rules are also used by Mattei~\textit{et al.}~\cite{mattei2013bribery}.

Interestingly, it is \NP-hard to determine the winner for the voting rule One-Step-$k$-Approval ($\OK$) if~$k$ is part of the input. Therefore, we restrict our analysis to efficient cases of $\OK$ where $k$ has a value which is polynomial in $n$ and $m$ or for $\mathcal{O}$-legal profiles where $k$ is a power of $2$. We denote these cases by $\OKeff$. Here, the winner can be determined in polynomial time using results by Brafman~\etal~\cite[Theorem 9]{brafman2010finding} and Mattei~\etal~\cite[Lemma 1]{mattei2013bribery}:

\begin{enumerate}
	\item If $k$ is polynomial in $n$ and $m$, then by a result of Brafman~\etal~\cite[Theorem 9]{brafman2010finding}, for a given acyclic CP-net, an order of the issues, and a candidate, it is possible to find the next best candidate in polynomial time. Hence, for every voter, we can list the candidates ranked on the first~$k$ positions in polynomial time and output the candidate with the maximum score.
	\item For $k = 2^j$ with $j \in \mathbb{N}$ and a global order of the issues, Mattei~\etal~\cite[Lemma 1]{mattei2013bribery} state that only the first $m-j$ issues are relevant and, more drastically, the last $j$ issues can be removed because each voter can only rearrange the first $2^j$ candidates with changes in his CP-net. So all of these candidates are the same on the first $m-j$ issues and cover all possible values for the last $j$ issues while all of them get a vote. \footnote{We note that while it is not possible to explicitly list all winners in this case, we can output the set of winners with the use of ``don't care'' wildcards for the last $j$ issues in polynomial space}
\end{enumerate}

For the general case, Theorem \ref{thm:evaluateokarbitraryk} proves \NP-hardness for just evaluating a given voting-scheme. In the rest of this work we will focus on $\OKeff$ instead of $\OK$.

\begin{theorem}\label{thm:evaluateokarbitraryk}
It is \NP-hard to determine the winner of a given profile when using the voting rule One-Step-$k$-Approval with arbitrary $k$ of polynomial size in $n$ and $m$ (i.e., exponential value).
\end{theorem}

\begin{proof}
We reduce \DreiSAT{} on $\satm$ clauses and $\satn$ variables to an $OK$-election with $\satm$ voters, $\satn$ issues and $k = 7 \cdot 2^{\satn-3}$. So given a formula $F$ on variables $\satv_1, \ldots, \satv_\satn$ and clauses $\satc_1, \ldots, \satc_\satm$ we introduce for each variable $\satv_i$ the issue $\Issue_i$. For each clause $\satc$ containing the variables $\satv_i, \satv_j, \satv_{\ell}$, we create one voter ordering the issues by $(\Issue_i > \Issue_j > \Issue_{\ell} > \hat{\Issue} \setminus \{\Issue_i, \Issue_j, \Issue_k\})$ with $\hat{\Issue}$ being the set of all issues. This voter prefers the three issues in a way such that the two non-satisfying combinations are on the last position (regarding only these three issues). The assignment and order of the last $\satn-3$ issues is not relevant. Therefore the voter votes for the first $7 \cdot 2^{\satn-3}$ candidates which exactly corresponds to all candidates with a satisfying combination for the clause $\satc$ (no matter what the rest of the issues are set to). 
Now a winner of the election obtains $\satm$ votes if and only if $F$ is satisfiable.
\end{proof}

\paragraph*{Bribery}
We consider the problem that an external agent, the \emph{briber}, who knows the CP-nets of all voters, asks them to execute changes in their cp-statements. We distinguish the cases that the briber can ask for a change in cp-statements of independent issues only ($\IV$), dependent issues only ($\DV$), or in all ($\IDV$) issues~\cite{mattei2013bribery}.   
We consider the following five cost schemes~\cite{mattei2013bribery}:
 \begin{itemize}
 \item $\Cequal$: Any amount of change in a single CP-net has the same unit cost.
 \item $\Cflip$: The cost of changing a CP-net is the total number of individual cp-statements that must be flipped to obtain the desired change.
 \item $\Clevel$: The cost of a bribery is computed\footnote{The formula given here differs from the one of Mattei~\etal~\cite{mattei2013bribery}. See the argument of Dorn and Krüger~\cite[Remark~1]{dorn2014hardness} why both are equivalent.} as $\sum_{X_j \in M'} (k+1-\text{level}(X_j))$, where $M'\subseteq M$ is the set of bribed issues for this voter, $k$ is the number of levels in the CP-net, and $\text{level}(X_j)$ corresponds to the depth of issue $X_j$ in the dependency graph. More precisely, \\ $\text{level}(X_j)= \begin{cases} 1 & \text{if }  X_j \text{ is an independent issue} \\ i+1 & \text{else, with } i = max\{ \text{level}(\Issue_k)~|~\Issue_k \text{ is a parent of }\Issue_j\}.\end{cases}$
 \item $\Cany$: The cost is the sum of the flips, each weighted by a specific cost.
 \item $\Cdist$: This cost scheme requires a fixed order of the issues for each voter (not necessarily the same for each of them), inducing a  strict total order over all candidates. The cost to bribe a voter to make~$c$ his top candidate is the number of candidates which are better ranked than $c$ in this order.
\end{itemize}
\noindent
Additionally, these cost schemes are extended by a cost vector ${\bf Q} \in (\mathbb{N})^n$ for an individual cost factor for each voter. The factor for voter~$v_i$ is denoted by $Q[i]$ and is multiplied with the costs calculated by the used cost scheme to obtain the amount that the briber has to pay to bribe $v_i$. 

\noindent The destructive $(D,A,C)$-bribery problem is then defined in the following way: 
\begin{quote}
$(D,A,C)$-\textsc{Destructive-Bribery} (\destBrib)\newline
\textbf{Given:}  A profile of $n$ CP-nets over $m$ common binary issues,  a winner determination voting rule $D \in \{\SM, \OP, \OV, \OKeff\}$, a cost scheme $C\in \{\Cequal,\Cflip,\Clevel,\Cany,\Cdist\}$, a bribery action $A\in \{\IV, \DV,$ $\IDV\}$, a cost vector ${\bf Q}\in (\mathbb{N})^n$, a budget~$\beta \in \mathbb{N}$, and a disliked (`hated') candidate~$h$. 
With $\SM$, we also require $\mathcal{O}$-legality for one common order and with $\Cdist$, and $\OKeff$ up to $n$ given total orders over the issues. \newline
\textbf{Question:} Is it possible to change the cp-statements of the voters such that the candidate $h$ is not in the set of winners of the bribed election, without exceeding~$\beta$?
\end{quote}

We also consider \textsc{Weighted-$(D,A,C)$-}\destBrib, which is defined in the same way, but with weighted voters, which is a typical variant for bribery problems (see the overview of Faliszewski \textit{et al.}~\cite{faliszewski2009hard}).
 
Moreover, we also consider weighted and unweighted $(D,A,C)$-\negative-\destBrib. The notion of negative bribery was introduced by Faliszewski~\etal\cite{faliszewski2009hard} for the constructive case (to make a candidate win the election) in order to cover a more inconspicuous way of bribery: the briber wants to make his preferred candidate~$p$ win by not bribing any voter to vote directly for~$p$, therefore just redistributing the votes for the other candidates through bribery. For the destructive case we consider in this work, the analogue restriction is to prohibit bribing voters to vote against the disliked candidate if they have not done so before (recall that with $\OK$ and $\OV$, a voter votes for several candidates).

Sometimes we show that a result holds for both, the negative and the non-negative case. We indicate this by $\langle\negative\rangle$ in the problem name.

For all our hardness results we prove only \NP-hardness for the corresponding problems, but immediately obtain \NP-completeness due to obvious membership in \NP for all of the problems.

\begin{table}[ht]
\caption{Complexity results (\Pclass{} stands for solvability in polynomial time, \NP-c for \NP-completeness) for variants of the destructive bribery problem in CP-nets shown in this paper. These variants are specified by a cost scheme~($\Cequal,\Cflip,\Clevel,\Cany,\Cdist$), given at the top of the corresponding column, and a voting scheme~($\SM,\OP,\OV,\OKeff$) at the beginning of the corresponding row. The unweighted variants are given in the top half of the table, the weighted ones are listed in the bottom half. The given results all hold for the bribery actions $IV$, $DV$, and $IV+DV$, if not stated differently. The cases that are solvable in polynomial time, if the entry in the cost vector is identical for every voter, are not included. 
\label{table:results}}

\centering
{\begin{footnotesize}\tabcolsep=0.18cm
\begin{tabular}{rlrccccccc}
\toprule 
 & & & ~$\Cequal$~ & ~$\Cflip$~ & ~$\Clevel$~ & ~$\Cany$~ & ~$\Cdist$~ & ~negative~ & non-negative\\
\cmidrule{2-8}
\parbox[t]{3mm}{\multirow{8}{*}{\rotatebox[origin=c]{90}{unweighted}}}  & $\SM$ & & \PResult{thm:negative:sm} & \PResult{thm:negative:sm} & \PResult{thm:negative:sm} & \PResult{thm:negative:sm} & \PResult{thm:negative:sm} & \texttt{Thm.\ref{thm:negative:sm}} & \texttt{Thm.\ref{thm:negative:sm}}\\
   & $\OP$ & $\IV$ & \NPcResult{thm:negativebase} & \NPcResult{thm:negativebase} & \NPcResult{thm:negativebase} & \NPcResult{thm:negativebase} & \NPcResult{thm:negativebase} & \texttt{Thm.\ref{thm:negativebase}} & \texttt{Cor.\ref{cor:positivebase}}\\
 & & $\DV$ & \NPcResult{thm:negativebase} & \NPcResult{thm:negativebase} & \NPcResult{thm:negativebase} & \NPcResult{thm:negativebase} & \NPcResult{thm:negativebase} & \texttt{Thm.\ref{thm:negativebase}} & \texttt{Cor.\ref{cor:positivebase}}\\
 & & \hspace*{-2cm} $\IDV$ & \PResult{thm:negativebase:p} & \NPcResult{thm:negativebase} & \NPcResult{thm:negativebase} & \NPcResult{thm:negativebase} & \NPcResult{thm:negativebase} & \texttt{Thm.\ref{thm:negativebase:p}/\ref{thm:negativebase}} & \texttt{Cor.\ref{cor:positivebase:p}/\ref{cor:positivebase}}\\
 & $\OV$ & & \PResult{thm:negative:ov} & \PResult{thm:negative:ov} & \PResult{thm:negative:ov} & \PResult{thm:negative:ov} & \PResult{thm:negative:ov} & \texttt{Thm.\ref{thm:negative:ov}} & \texttt{Thm.\ref{thm:negative:ov}}\\
 & $\OKeff$ & $\IV$ & \NPcResult{thm:negativebase} & \NPcResult{thm:negativebase} & \NPcResult{thm:negativebase} & \NPcResult{thm:negativebase} & \NPcResult{thm:negativebase} & \texttt{Thm.\ref{thm:negativebase}} & \texttt{Cor.\ref{cor:positivebase}}\\
 & & $\DV$ & \NPcResult{thm:negativebase} & \NPcResult{thm:negativebase} & \NPcResult{thm:negativebase} & \NPcResult{thm:negativebase} & \NPcResult{thm:negativebase} & \texttt{Thm.\ref{thm:negativebase}} & \texttt{Cor.\ref{cor:positivebase}}\\
 & & \hspace*{-1.5cm} $\IDV$ & \PResult{thm:negativebase:p} & \NPcResult{thm:negativebase} & \NPcResult{thm:negativebase} & \NPcResult{thm:negativebase} & \NPcResult{thm:negativebase} & \texttt{Thm.\ref{thm:negativebase:p}/\ref{thm:negativebase}} & \texttt{Cor.\ref{cor:positivebase:p}/\ref{cor:positivebase}}\\
\cmidrule{2-8}
\parbox[t]{3mm}{\multirow{6}{*}{\rotatebox[origin=c]{90}{weighted}}} & $\SM$ & & \NPcResult{thm:neg:weig:SM} & \NPcResult{thm:neg:weig:SM} & \NPcResult{thm:neg:weig:SM} & \NPcResult{thm:neg:weig:SM} & \NPcResult{thm:neg:weig:SM} & \texttt{Thm.\ref{thm:neg:weig:SM}} & \texttt{Thm.\ref{thm:neg:weig:SM}}\\
 & $\OP$ & & \NPcResult{thm:negative:weighted:OPOK:npc} & \NPcResult{thm:negative:weighted:OPOK:npc} & \NPcResult{thm:negative:weighted:OPOK:npc} & \NPcResult{thm:negative:weighted:OPOK:npc} & \NPcResult{thm:negative:weighted:OPOK:npc} & \texttt{Thm.\ref{thm:negative:weighted:OPOK:npc}} & \texttt{Thm.\ref{thm:negative:weighted:OPOK:npc}} \\
 & $\OV$ & $\IV$ & \NPcResult{thm:weig:OVmain} & \NPcResult{thm:weig:OVmain} & \NPcResult{thm:weig:OVmain} & \NPcResult{thm:weig:OVmain} & \NPcResult{thm:weig:OVmain} & \texttt{Thm.\ref{thm:weig:OVmain}} & \texttt{Thm.\ref{thm:weig:OVmain}}\\
 & & $\DV$ & \NPcResult{thm:weig:OVmain} & \NPcResult{thm:weig:OVmain} & \NPcResult{thm:weig:OVmain} & \NPcResult{thm:weig:OVmain} & \PResult{thm:neg:weig:OVDVCdist:P} & \texttt{Thm.\ref{thm:weig:OVmain}/\ref{thm:neg:weig:OVDVCdist:P}} & \texttt{Thm.\ref{thm:weig:OVmain}/\ref{thm:neg:weig:OVDVCdist:P}}\\
 & & \hspace*{-2cm} $\IDV$ & \NPcResult{thm:weig:OVmain} & \NPcResult{thm:weig:OVmain} & \NPcResult{thm:weig:OVmain} & \NPcResult{thm:weig:OVmain} & \NPcResult{thm:weig:OVmain} & \texttt{Cor.\ref{thm:weig:OVmain}} & \texttt{Thm.\ref{thm:weig:OVmain}}\\
 & $\OKeff$ & & \NPcResult{thm:negative:weighted:OPOK:npc} & \NPcResult{thm:negative:weighted:OPOK:npc} & \NPcResult{thm:negative:weighted:OPOK:npc} & \NPcResult{thm:negative:weighted:OPOK:npc} & \NPcResult{thm:negative:weighted:OPOK:npc} & \texttt{Thm.\ref{thm:negative:weighted:OPOK:npc}} & \texttt{Thm.\ref{thm:negative:weighted:OPOK:npc}} \\
\cmidrule{2-8}
\bottomrule
\end{tabular}
\end{footnotesize}}
\end{table}

\newcommand{\wine}{\textit{wine}}
\newcommand{\beer}{\textit{beer}}
\newcommand{\fish}{\textit{fish}}
\newcommand{\meat}{\textit{meat}}
\newcommand{\chips}{\textit{chips}}
\newcommand{\rice}{\textit{rice}}
\newcommand{\main}{\textit{main}}
\newcommand{\side}{\textit{side}}
\newcommand{\drink}{\textit{drink}}

\begin{table}[ht]
\caption{Example for three CP-nets of the voters Alice, Bob, and Charlie over the three issues \emph{\main{} dish}, \emph{\side{} dish}, and \emph{\drink} with the domains $D(\main)=\{\fish,\meat\}$, $D(\side) = \{\rice,\chips\}$, and $D(\drink) = \{\wine,\beer\}$.  \label{table:example}}
\setlength{\tabcolsep}{4pt}
\centering
\begin{small}
\begin{tabular}{llll} 
\toprule 
 & \main & \side & \drink \\
\cmidrule{2-4}
Alice (A)	& $\fish > \meat$	& $\rice > \chips$	& $\meat\colon \beer > \wine$ \\
~($\main>\drink>\side$)		&					&					& \hspace*{.21cm}$\fish\colon \wine > \beer$\\
\cmidrule{2-4}
Bob (B)		& $\fish > \meat$	& $\chips > \rice$	& $\meat,\chips\colon \beer > \wine$ \\
~($\main>\side>\drink$)		& 					& 					& \hspace*{.21cm}$\meat,\rice\colon \beer > \wine$ \\
		& 					& 					& \hspace*{.21cm}$\fish,\chips\colon \beer > \wine$ \\
				& 					& 					& \hspace*{.42cm}$\fish,\rice\colon \wine > \beer$ \\
\cmidrule{2-4}
Charlie (C)	& $\chips\colon \fish > \meat$	& $\chips > \rice$	& $\chips\colon \beer > \wine$ \\
~($\side>\drink>\main$)		& \hspace*{.21cm}$\rice\colon \meat > \fish$ 	&					& \hspace*{.21cm}$\rice\colon \wine > \beer$\\
\bottomrule
\end{tabular}\end{small}
\end{table}

\paragraph*{Example}
Table~\ref{table:example} shows the CP-nets of the voters Alice (A), Bob (B), and Charlie (C) over the three issues \emph{\main{} dish}, \emph{\side{} dish}, and \emph{\drink} with the domains $D(\main)=\{\fish,\meat\}$, $D(\side) = \{\rice,\chips\}$, and $D(\drink) = \{\wine,\beer\}$ to vote for a joined meal. Additionally, the individual orders of the issues are given for each voter, implying the following total orders as
expansions of the partial orders defined by the CP-nets:
\begin{center}{\setlength{\tabcolsep}{1pt}
\begin{tabular}{rcccccccc}
A:~ & (\fish,\rice,\wine) & $>$ & (\fish,\chips,\wine) & $>$ & (\fish,\rice,\beer) & $>$ & (\fish,\chips,\beer) & $\dots$\\
B:~ & (\fish,\chips,\beer) & $>$ & (\fish,\chips,\wine) & $>$ & (\fish,\rice,\wine) & $>$ & (\fish,\rice,\beer) & $\dots$\\
C:~ & (\fish,\chips,\beer) & $>$ & (\meat,\chips,\beer) & $>$ & (\fish,\chips,\wine) & $>$ & (\meat,\chips,\wine) & $\dots$\\
\end{tabular}}
\end{center}

Using the voting rule $\OK$ with $k = 3$, the candidate (\fish,\chips,\wine)
wins the election, because it is the only one receiving a point from
each of the three voters. With the voting rule $\OP$, the candidate
(\fish,\chips,\beer) is the winner, thanks to the two points from Charlie and Bob. The same candidate wins with the voting rule $\SM$ with respect to the order~$\mathcal{O}\colon \side > \main > \drink$, for which the given profile is $\mathcal{O}$-legal. In the majority vote on the issue \side{}, Bob
and Charlie prefer \chips{}, so \chips{} is chosen as a \emph{side dish}. Because of this Charlie votes---like the other two voters---for \fish{} in the majority vote for the second issue of order~$\mathcal{O}$. Finally Alice gets outvoted in the last issue, so the candidate (\fish,\chips,\beer) is the winner with the voting rule $\SM$, too. With the voting rule $\OV$, Alice casts her veto against (\meat,\chips,\wine), Bob
casts his veto against (\meat,\rice,\wine), and Charlie casts his veto against (\fish,\rice,\beer).
Therefore, the remaining five
candidates are the winning candidates with this rule.
If a unique winner was
needed, a tie-breaking rule could be used.

If the briber wanted to prevent candidate $h = (\fish,\chips,\beer)$ from
winning with voting rule $\OP$, it would be sufficient to bribe Bob to
flip his preference in issue \side{} to $\rice > \chips$. This would make
(\fish,\rice,\wine) the top candidate of Bob and, since Alice is voting for the very same candidate in the first place, the winner of the election. So the briber can reach his goal by this bribery. Note that this flip is only possible with the bribery action $\IV$ or $\IDV$, because \side{} is an independent issue for Bob.

How much does the briber have to pay for this requested flip? For the cost
scheme~$\Cany$, this directly depends on the input, since each flip can
have its own costs. With~$\Cflip$, the cost factor is 1, since only one cp-statement has to be flipped. With~$\Clevel$, it is 2 for this flip, because Bob’s CP-net has two levels and the issue \side{} is an independent one. The costs are the same for~$\Cdist$, because Bob only prefers the candidates (\fish,\rice,\beer) and (\fish,\chips,\wine) to (\fish,\rice,\wine). Finally, with~$\Cequal$, bribing Bob has the same prize (assuming equal cost vector entries) as bribing any other arbitrary voter with a different top candidate to vote for (\fish,\rice,\wine). To obtain the final costs the briber has to pay, each of these values is then multiplied by the corresponding entry of the cost vector~${\bf Q}$. Therefore, the briber might sometimes be cheaper off to bribe a voter with a small cost vector entry to flip a lot of cp-statements, than one with only a few flips required but having a huge cost vector entry.

\section{The unweighted case}
In this section, we investigate the case where voters are unweighted. 
\begin{theorem}\label{thm:negative:sm}
$(SM, A, C)$-$\langle\negative\rangle$-\destBrib with bribery action $A\in\{IV, DV, IV+DV\}$ and cost scheme $C\in\{\Cequal,\Cany,\Cflip,\Clevel,\Cdist \}$ is solvable in polynomial time.
\end{theorem}

\begin{proof}
We start with the negative case. For each issue find the minimum costs to spend for reaching a majority against~$h$. This can be done by collecting all costs for one issue, sorting and summing up. Finally the issue which is cheapest to bribe is chosen. This works for each cost scheme for which it is easy to calculate the bribery costs for a single flipped issue, which is the case for all cost schemes used here\footnote{This is most unintuitive for $\Cdist$, but identifying the top candidate after bribing one issue and determining the respective cost can be done in polynomial time as described by Mattei~\etal~\cite[Theorem~3]{mattei2013bribery}.}.

Note that depending on the allowed bribery action, not every voter may be bribable in each issue. Similarly,  we have to ignore voters who initially vote for~$h$ but who do not vote for~$h$ any more after a bribery of the considered issue.
These voters have to be taken into account in the non-negative case, though. However, this is the only modification needed for this case. 
\end{proof}

\begin{theorem}\label{thm:negativebase}
($D,A,C$)-\negative-\destBrib{} with $D\in\{\OP, \OKeff\}$, $A\in\{\IV,\DV\}$, $C\in\{\Cequal,\Cflip,\Clevel,\Cany,\Cdist\}$ is \NP-complete. In addition, ($D,\IDV,C'$)-\negative-\destBrib{} is \NP-complete for $C'\in\{\Cflip,\Clevel,\Cany,\Cdist\}$. 
All these results hold even if all entries in the cost vector are identical.
\end{theorem}

\begin{proof}
We give a base reduction from \NAEdreiSAT{} to prove Theorem~\ref{thm:negativebase}. For this reduction we claim that some properties hold, which we will then show to hold for the various combinations of allowed bribery action, cost scheme and voting rule.

\noindent Assume we are given an \NAEdreiSAT{} instance with $\satm{}$ clauses and $\satn{}$ variables. For each variable~$\satv_i$, we create one issue $\Issue_i$. Since each issue has exactly two possible assignments, we can establish a one-to-one relation between the full assignments of the issues $\Issue_i$ and the one of the variables $\satv_i$. For this relation we say the assignment of $\issue_i$ to $\Issue_i$ corresponds to the assignment of $1$ to variable $\satv_i$.  
We will later on---in the extensions of the base case---create additional gadget-issues. 

For each clause $\satc$ with the variables $\satv_q,\satv_r,\satv_s$ and each of the six different satisfying assignments to these variables for $\satc$, we create one voter with the following preferences: He prefers $\issue_i$ over $\overline{\issue_i}$ for each issue with $i\notin\{q,r,s\}$. For the remaining three issues he prefers $\issue_l$ over $\overline{\issue_l}$, if $1$ is assigned to variable $\satv_l$ in the satisfying assignment of $\satc$; he prefers $\overline{\issue_l}$ over $\issue_l$ otherwise. For the considered clause with variables $\satv_q,\satv_r,\satv_s$, we obtain 6 voters which we refer to as \textit{qrs}-voters. Doing this for all clauses, we obtain $6\satm$ voters. Finally, we create $\satm-1$ additional voters, all having~$h$ as their top-candidate\footnote{The candidate $h$ can correspond to an arbitrary solution to the formula because \textsc{AnotherSAT} (the variant of \textsc{SAT} where, given a formula and a satisfying assignment for it, one has to find another satisfying assignment for the formula) still is \NP-complete, see for example \cite{Juban1999}.}. We set the entry in the cost vector for each voter to $1$.

\noindent We assume that the following property holds:

\begin{enumerate}[noitemsep,label={(\roman*)}]
\item No voter who is voting for $h$ can be bribed to vote against~$h$.
\end{enumerate}
We further assume that for all other voters the following properties apply:
\begin{enumerate}[noitemsep,label={(\roman*)}]\setcounter{enumi}{1}
\item The preferences within issues associated with the clause the voter was created for cannot be changed.
\item Changing the preferences within a gadget-issue does not help the briber.
\item The preferences within all remaining issues can be bribed freely.
\end{enumerate}

\noindent Assuming that properties (i) to (iv) are fulfilled, it is easy to see that the bribery instance can be solved if and only if the corresponding \NAEdreiSAT-formula is satisfiable. Given a satisfying assignment to the formula, we can translate it to the winning candidate by following the one-to-one relation. Since the assignment satisfies each clause, the briber can bribe one of the six voters for each clause to vote for the winning candidate (following (ii) and (iv)). Therefore this candidate will have $\satm$ votes in the end, while there are still only $\satm-1$ votes for $h$. The other direction can be shown analogously with the help of property (iii).

We will now show the extensions to this base reduction to prove Theorem~\ref{thm:negativebase}. Property~(i) is fulfilled because we are looking at the negative case.

($\OP,\IV,C$): We need one gadget-issue $\Issue^*$. Each voter is set to prefer $\issue^*$ over $\overline{\issue^*}$. For each \textit{qrs}-voter, the issues $\Issue_q,\Issue_r,\Issue_s$ are changed to depend on $\Issue^*$, keeping their original preferences for $\issue^*$ and inverted preferences for $\overline{\issue^*}$. Here we utilize that \NAEdreiSAT{} is closed under complement, therefore the assignment of $\Issue^*$ is not important at all. This modification ensures the properties (ii)--(iv) to hold for each cost scheme, since no costs are involved.

($\OP,\DV,C$): Once again we need one gadget-issue $\Issue^*$. Each voter is set to prefer $\issue^*$ over $\overline{\issue^*}$. Complementary to the case before, for each \textit{qrs}-voter, each issue $\Issue_i$ with $i\notin\{q,r,s\}$ is changed to depend on $\Issue^*$, keeping their original preferences for $\issue^*$ and inverted preferences for $\overline{\issue^*}$. This modification ensures the properties (ii)--(iv) to hold for each cost scheme, since no costs are involved.

($\OP,\IDV,\Cany$): With $\IDV$ we need costs and an appropriate budget to ensure that the issues $\Issue_q,\Issue_r,\Issue_s$ cannot be bribed for a \textit{qrs}-voter. With $\Cany$ we can simply set the costs to bribe these issues to~$1$, and for each remaining issue to $0$. With the budget set to $\budget=0$ this ensures the properties (ii)--(iv) hold.

($\OP,\IDV,\Cflip$): We need to add $\satm^2(\satn-3)$ gadget-issues, which we denote as $\Issue^*_{a,b}$ with $1\leq a\leq \satm$ and $1\leq b \leq \satm(\satn-3)$. For each \textit{qrs}-voter, each gadget-issue $\Issue^*_{j,b}$ with $1\leq b \leq \satm(\satn-3)$ depends on the issues $\Issue_q,\Issue_r,\Issue_s$ corresponding to the variables of the clause $\satc_j$. With the most preferred assignment in these three issues the preference in the gadget-issue is set to $\issue^*_{j,b} > \overline{\issue^*_{j,b}}$, and for each other assignment to $\overline{\issue^*_{j,b}} > \issue^*_{j,b}$. Finally we set the budget to $\budget=\satm(\satn-3)$. This ensures the properties (ii)--(iv) to hold.

($\OP,\IDV,\Clevel$): We need $\satm\satn$ gadget-issues $\Issue^*_{b}$ for $b \in \{1, \dots, \satm \satn\}$. In contrast to $\Cflip$, these issues do not depend on the issues $\Issue_q,\Issue_r,\Issue_s$ corresponding to the variables of the clause $\satc_j$ in \emph{parallel}, but in a \emph{queue}. 
So $\Issue^*_{1}$ depends on the issues $\Issue_q,\Issue_r,\Issue_s$ for a \textit{qrs}-voter. 
Only for the most preferred assignment within these three issues we set the preference of issue $\Issue^*_{1}$ to $\issue^*_{1} > \overline{\issue^*_{1}}$; in all other cases we set it to $\overline{\issue^*_{1}} > \issue^*_{1}$. 
For each subsequent issue $\Issue^*_{b}$ of this queue with $2\leq b \leq \satm\satn$, we set the preferences $\issue^*_{b-1} : \issue^*_{b} > \overline{\issue^*_{b}}$ and $\overline{\issue^*_{b-1}} : \overline{\issue^*_{b}} > \issue^*_{b}$. In addition, for all issues $\Issue_i$ with $i \in \{1, \dots, \satn\}\setminus\{q,r,s\}$, we set $\issue^*_{\satm\satn}: \issue_i > \overline{\issue_i}$ and else $\overline{\issue_i} > \issue_i$. Finally we set the budget to $\budget=\satm(\satn-3)$. This ensures the properties (ii)--(iv).

($\OP,\IDV,\Cdist$): We add $\lceil\log\satm\rceil + 1$ gadget-issues, denoted by $\Issue^*_a$. Every voter prefers $\issue^*_a > \overline{\issue^*_a}$ for each such issue. For each voter we set the (not necessarily identical) order as follows. The most important issues are the three issues corresponding to the variables of the clause $\satc_j$ associated with it, followed by the gadget-issues, and then by the remaining $\satn-3$ issues. The exact order within these three blocks is not important. We set the budget to $\budget=\satm \cdot 2^{\satn-3}-1$. This ensures that the briber can bribe all of the least important $\satn-3$ issues for~$\satm$ voters, while the budget is still too low to bribe even only one of the three most important issues of just one voter. Note that bribing such an issue costs at least $2^{n-3 + \lceil\log{m}\rceil} > \budget$. Therefore the properties (ii)--(iv) hold.

Since $\OP$ is the special case of $\OKeff$ with $k=1$, the \NP-completeness results shown so far carry over to $\OKeff$.
\end{proof}

We can achieve the same results for the non-negative cases, but for this we have to drop the constraint that these problems are \NP-complete even if the cost vector contains only 
identical values.

\begin{corollary}\label{cor:positivebase}
$(D, A, C)$-\destBrib is \NP-complete for each combination of a voting rule $D \in \{OP, \OKeff\}$, a cost scheme $C\in\{\Cequal,\Cflip,\Clevel,\Cany,\Cdist\}$ and a bribery action $A\in\{IV,DV\}$. Moreover, $(D, IV+DV, C')$-\destBrib with cost scheme $C' \in \{\Cflip,\Clevel,\Cany,\Cdist\}$ is \NP-complete, too.
\end{corollary}

\begin{proof}
This follows by the same techniques used in the proof of Theorem \ref{thm:negativebase}. The only difference is that property (i) is not automatically satisfied. We can achieve unbribability of the voters voting for $h$ by setting the entries of the cost vector for these voters to $\budget+1$. Therefore, \NP-completeness follows.
\end{proof}

For the two remaining cases of $\OP$ and $\OKeff$, we will show solvability in polynomial time.

\begin{theorem}\label{thm:negativebase:p}
$(D, \IDV, \Cequal)$-\negative-\destBrib is solvable in polynomial time with $D \in \{\OP, \OKeff\}$.
\end{theorem}

\begin{proof}
We partition the set of voters $V$ into the set $V_h$ of voters casting a vote to $h$ and the remaining voters $V\setminus V_h$.

Let us start with the voting rule $\OKeff$ and the case that $k$ is polynomial in $n$ and $m$ (denoted by $\textit{poly}(n,m)$). For this we show that the set of candidates who can be made a winner without exceeding the budget, is small enough to try out each of them. Since the voters in $V_h$ can only be bribed in the least important $\lceil\log k\rceil$ issues, the number of candidates these voters can be bribed to vote for is $2^{\lceil\log k\rceil}=\textit{poly}(n,m)$. Each of these we handle as a potential winning candidate. Because the voters of $V\setminus V_h$ can be bribed to vote for any of the exponentially many candidates ($IV+DV$) for the same costs ($\Cequal$) it is sufficient to take only those candidates into consideration who initially get at least one vote by a voter of $V\setminus V_h$. Since there are at most $kn$ such candidates, we can handle them as potential winning candidates, too. For each of the only $\textit{poly}(n,m)$ many potential winning candidates, we calculate the bribing costs for each voter to vote for this candidate, we sort them accorting to the costs, and add the bribing costs for as many voters as needed (taking the cheapest ones first). Each of these steps can be done in polynomial time in $n$ and $m$, implying an overall polynomial running time.

The remaining case of $\OKeff$ -- with $k$ being a power of $2$ and a global order over the issues for all voters given -- coincide with $\OP$ when the least important $\log k$ issues are removed. This is due to the fact that a bribery of any subset of those least $\log k$ issues for any voter only permutes the set of candidates this voter votes for. We can solve this again by identifying the potential winners and calculating the costs for the required bribery. Since with $\OP$ every voter votes just for one candidate, there are only up to $n$ such potential winning candidates. The bribery costs for letting them win the election can be calculated by the algorithm introduced by Dorn and Kr\"uger~\cite[Theorem~$10$]{dorn2014hardness} for solving the constructive case in polynomial time. 
\end{proof}

The non-negative case can be solved by a very similar algorithm.

\begin{corollary}\label{cor:positivebase:p}
$(D, \IDV, \Cequal)$-\destBrib is solvable in polynomial time with $D \in \{\OP, \OKeff\}$. 
\end{corollary}

\begin{proof}
Corollary~\ref{cor:positivebase:p} can be shown by small adjustments to the proof of Theorem~\ref{thm:negativebase:p}, where the analogous negative cases are covered. For voting rule $\OKeff$ and the case that $k$ is polynomial in $n$ and $m$, the briber is now able to bribe the voters in $V_h$ freely, therefore we have to build a second similar bribery-costs-sorted list of voters in $V_h$. Summing up the cheapest of those costs to make a specific candidate a winner of the election is a bit more complicated, since a bribery of a voter in $V_h$ could change this voter to not vote for $h$ any longer. But this can be handeled in time $\textit{poly}(n,m)$ easily. As potential winning candidates we try again each candidate which did get at least one vote initially and additionally we try the candidate which has a complementary assignment to $h$ in each issue, because each voter can be bribed to vote for him but not for $h$.

The remaining case of $\OKeff$ -- with $k$ being a power of $2$ and a global order over the issues for all voters given -- coincides with $\OP$ again and the same arguments as used in the proof of Theorem~\ref{thm:negativebase:p} can be used. This time we can use the result of Mattei~\etal~\cite[Theorem~$8$]{mattei2013bribery} for solving the constructive, non-negative case. 
\end{proof}

Interestingly, the voting rule $\OV$ is another special case of $\OK$ which can be evaluated and solved in polynomial time. 

\begin{theorem}\label{thm:negative:ov}
$(OV, A, C)$-$\langle\negative\rangle$-\destBrib  with bribery action $A\in\{IV, DV, IV+DV\}$ and cost scheme $C\in\{\Cequal, \Cflip, \Clevel, \Cany, \Cdist\}$ is in \Pclass{}.
\end{theorem}

\begin{proof}
To solve the negative version we distinguish two cases. First, in the case $n<2^m$ there has to be at least one candidate who did not obtain a veto, by the pigeonhole principle. If $h$ gets at least one veto, it is a yes-instance. Otherwise it would be sufficient to bribe the cheapest voter to cast his veto to $h$. But as this is not allowed, it is a no-instance.

Second, if $n\geq 2^m$, the number of candidates is linear in the input size, so one can calculate the costs to let a candidate win for each candidate by applying the polynomial-time algorithm by Dorn and Krüger~\cite[Theorem 12]{dorn2014hardness} for the constructive negative $(OV, A, C)$-\textsc{Bribery} case which is formulated for the co-winner case but can easily be adapted for the unique-winner case. The bribery with overall minimum costs can then be applied if the budget is sufficient.

The different bribery actions can just prevent the briber to bribe specific voters in some issues, therefore the theorem holds for each bribery action. For the non-negative version, we note that in the first case with $n < 2^m$, it is now sufficient (and possible) to bribe the cheapest voter to cast his veto against~$h$. The rest remains the same.
\end{proof}

Due to space constraints we omit further examples of special cases of $\OKeff$ which are solvable in polynomial time. We just remark that the combinations of $\OP$, $\IDV$, and $\Cflip,\Clevel,\Cdist$, and $\OP,\IV,\Cflip$, each for the unweighted, non-negative case and if all entries of the cost vector are the same, can be solved in polynomial time, too (see Section \ref{sec:CostVectorSameValue}). This is in line with the observations of Mattei~\etal~\cite[Theorem 7]{mattei2013bribery} and Dorn and Krüger~\cite[Theorem 7]{dorn2014hardness} that sometimes the bribery problem can be solved in polynomial time when the cost vector has only identical entries.

\section{The weighted case}
In this section, we consider the case of weighted voters, which turns out to be \NP-complete for almost all combinations---with two exceptions. The reductions we give here are all from the \Knapsack problem.

\begin{theorem}\label{thm:negative:weighted:OPOK:npc}
\weigh-$(D, A, C)$-$\langle\negative\rangle$-\destBrib is \NP-complete with voting rule $D \in \{\OP, \OKeff\}$, bribery action $A\in\{\IV,\DV,\IDV\}$, and cost scheme $C\in\{\Cequal, \Cflip, \Clevel, \Cany, \Cdist\}$.
\end{theorem}

\begin{proof}
For the voting rule $\OP$, we use a reduction from \Knapsack. Given a \Knapsack instance $(\{(w_1, v_1), \ldots, (w_n, v_n)\}, k, b)$ we construct a voting scheme such that a successful bribery against the candidate $h$ is possible if and only if the given Knapsack problem can be solved.

First, we use two issues $\Issue_1$ and $\Issue_2$, set $h = \overline{\issue_1}\,\overline{\issue_2}$, and the budget $\budget=b$. We create one voter preferring $h$ with weight $l+k-1$, where $l=\sum_{i=1}^n v_i$. We set the entry of the cost vector for this voter to $\budget+1$ in order to make him unbribable (just to make this proof hold for the non-negative case, too).

 For every object $(w_i, v_i)$ we add a voter preferring $\issue_1\overline{\issue_2}$ to all other candidates. This voter is weighted by $v_i$ and his entry in the cost vector is set to $w_i$. Last, we add a single voter of weight $l$ preferring $\issue_1\issue_2$ to all other candidates. Since this candidate should win the election, we make this voter unbribable by setting the entry in the cost vector to $\budget+1$.
 
For the bribery actions $\IV$ and $\IDV$ all issues are independent. For the bribery action $\DV$ we let the preferences of the issue $\Issue_2$ depend on the assignment of issue $\Issue_1$ for each voter created for an object $(w_i,v_i)$. Such a voter will therefore prefer $\issue_1$ independently over $\overline{\issue_1}$, and the preference for $\Issue_2$ will be $\issue_1:\overline{\issue_2}>\issue_2,~ \overline{\issue_1}:\issue_2>\overline{\issue_2}$.

By construction, only voters created for objects can be bribed, and only changing their favorite candidate to $\issue_1\issue_2$ is helpful here. Each bribery of one such voter results in a value of $1$ for each cost scheme, which is then multiplied by the entry of the cost vector. Note that a cost of $1$ is obvious in all cases except $\Cdist$. Here, we have to ensure that the target of our bribery is on the second position. As the prefered candidate is $\issue_1\overline{\issue_2}$ and our potential winner is $\issue_1\issue_2$, we have to set $\Issue_1$ before $\Issue_2$ in the order on the issues. The claimed cost of $1$ for this bribery action then holds. It is easy to see that there must exist a solution to the underlying \Knapsack instance in order to be able to prohibit the candidate~$h$ from winning. 

Since the briber cannot bribe the single voter voting for $h$ by construction, this reduction holds for the negative case as well. Because $\OP$ is a special case of $\OKeff$, this result automatically carries over.
\end{proof}

The main idea of this reduction can be adjusted to show \NP-completeness for the voting rules $\OV$ and $\SM$, too.

\begin{theorem}\label{thm:weig:OVmain}
\weigh-$(\OV, A, C)$-$\langle\negative\rangle$-\destBrib with a bribery action $A\in\{\IV,\DV,\IDV\}$ and cost scheme $C\in\{\Cequal, \Cflip, \Clevel, \Cany, \Cdist\}$ is \NP-complete, except for the combination $A=\DV$ and $C=\Cdist$.
\end{theorem}

\begin{proof}
We prove Theorem~\ref{thm:weig:OVmain} in a similar manner as Theorem~\ref{thm:negative:weighted:OPOK:npc} by reduction from \Knapsack. Let us start with the non-negative case. For this we again need the two issues $\Issue_1$ and $\Issue_2$, declare candidate $\overline{\issue_1}\,\overline{\issue_2}$ to be the disliked candidate~$h$, and set the budget to $\budget=b$. Let us start with the allowed bribery action $\IV$. First we create 2 voters, weighted by $2l$ with $l=\sum_{i=1}^n v_i$. One of them is casting his veto to candidate $\overline{\issue_1}\issue_2$, the other one to $\issue_1\overline{\issue_2}$. Furthermore we create one voter weighted by $l$ casting his veto to $c^*=\issue_1\issue_2$ and a last one of this kind weighted by $2l-2k+1$ vetoing for $h$. We fix all these four voters by setting their cost vector entry to $\budget+1$; this makes the dependencies of their issues unimportant.\\
Finally we create one voter for each object $(w_i,v_i)$ casting his veto to~$c^*$ with weight $v_i$ and set the entry of the cost vector for this voter to $w_i$. Since the briber can only bribe these object-voters, his only change consists in bribing voters with a total weight of at least $k$ to cast their veto to $h$ instead of to $c^*$. By doing so, $c^*$ will win with a score of $2l-k$, while $h$ will have a score of $2l-k+1$, and the last two candidates $2l$, each. Each such bribery will have a cost value of $1$ for the cost schemes $\Cequal,\Cflip,\Clevel,\Cdist$, and for $\Cany$ we can simply set the costs for the required flip to $1$. Together with the entries from the cost vector, a successful bribery can only be found within the given budget if and only if there exists a solution to the underlying \Knapsack instance.\\
Since there are no dependencies between the issues of the voters, this construction works for the bribery action $\IDV$, too. For $\DV$ we need to make some adjustments. Here we swap the role of the candidates $\issue_1\issue_2$ and $\overline{\issue_1}\issue_2$ and change the voters accordingly. The voters created for the objects $(w_i,v_i)$ have to be changed further, though. We set the preferences over issue $\Issue_1$ to be independent to $\issue_1>\overline{\issue_1}$ and for issue $\Issue_2$ to $\issue_1: \issue_2>\overline{\issue_2},~\issue_1:\overline{\issue_2}>\issue_2$. The weight and the entry in the cost vector of such a voter remain $v_i$ and $w_i$, respectively.
From here on, one can show the reduction by the same argumentation as before. 

This result can be deduced to the negative case by just some small adjustments of the weights of the unbribable voters. The only voter vetoing $h$ needs a weight of $l-k+1$, the unbribable voters vetoing $c^*$ are not needed any more, and the last two such voters vetoing the last two candidates need their weights to be lowered to $l$. So the briber has to bribe the voters casting their veto to one of the two latter candidates instead of $c^*$ (since changing them to $h$ is not allowed) to make $c^*$ the winner instead of $h$.
\end{proof}

\begin{theorem}\label{thm:neg:weig:SM}
\weigh-$(\SM, A, C)$-$\langle\negative\rangle$-\destBrib is \NP-complete for cost scheme $C\in\{\Cequal, \Cflip, \Clevel, \Cany, \Cdist\}$ and allowed bribery action $A\in\{\IV,\DV,\IDV\}$.
\end{theorem}

\begin{proof}
We prove Theorem~\ref{thm:neg:weig:SM} by reduction from \Knapsack.
For this we again need the two issues $\Issue_1$ and $\Issue_2$, declare candidate $\overline{\issue_1}\,\overline{\issue_2}$ to be the disliked candidate~$h$, and set the budget to $\budget=b$. We create a single voter preferring $h$, weighted by $2l$, with $l=\sum_{i=1}^n v_i$. We set the entry of the cost vector for this voter to $2\budget+1$ in order to make him unbribable (to make this proof hold for the non-negative case, too). We create one similar voter, weighted by $2l$, and a cost vector entry of $2\budget+1$, who prefers $\overline{\issue_1}\issue_2$. 

Because of the sequential fixing of assignments to the issues we swap the roles of the candidates $\issue_1\issue_2$ and $\issue_1\overline{\issue_2}$. So for every object $(w_i, v_i)$, we add a voter preferring $\issue_1\overline{\issue_2}$. This voter is weighted by $v_i$ and his entry in the cost vector is set to $w_i$. We create one additional voter preferring $\issue_1\overline{\issue_2}$, weighted by $l$ and with a cost vector entry of $2\budget+1$. Last, we add a single voter of weight $2l-2k+1$ preferring $\issue_1\issue_2$ and make this voter unbribable by setting the entry in the cost vector to $2\budget+1$. 

For the bribery actions $\IV$ and $\IDV$, all issues are independent. For the bribery action $\DV$, we need to make the preferences of the issue $\Issue_2$ depending on the assignment of issue $\Issue_1$ for each voter created for an object $(w_i,v_i)$. Such a voter will therefore prefer $\issue_1$ independently over $\overline{\issue_1}$, and the preference for $\Issue_2$ will be $\issue_1:\issue_2>\overline{\issue_2}$, $\overline{\issue_1}:\overline{\issue_2}>\issue_2$. In both variants we set the issue~$\Issue_1$ to be the leading one in the order~$\mathcal{O}$ needed to evaluate the voting rule~$\SM$. 

Again, the briber can only bribe those voters who were created for objects, and only changing their cp-statement $\overline{\issue_1}: \overline{\issue_2} > \issue_2$ to $\issue_1: \issue_2 > \overline{\issue_2}$ is helping here. Each bribery of one such voter will result in a value of $1$ for almost each cost scheme, which is then multiplied by the entry of the cost vector. The only exception is the cost scheme~$\Cdist$, where such a change would be free.  For this combination of cost scheme and bribery action we need a modification we will discuss later. 

This leads to the following votes in the sequential voting of $\SM$. Note that the assignment of $\Issue_1$ is fixed to the winner of the first round to all voters.

\begin{center}
\begin{tabular}{l|c|c}
 & ~without bribery~ & ~~~with bribery of total weight $\varphi$~~~ \\ \hline
$\issue_1$\, & $4l-2k+1$ & $4l-2k+1$ \\
$\overline{\issue_1}$ & $4l$ & $4l$ \\\hline
$\issue_2$ & $4l-2k+1$ & $4l-2k+1+\varphi$ \\
$\overline{\issue_2}$ & $4l$ & $4l-\varphi$ 
\end{tabular} 
\end{center}

So only with a bribery of voters with a total weight $\varphi \geq k$, the candidate $\overline{\issue_1}\issue_2$ will win the election. In every other case the candidate $h=\overline{\issue_1}\,\overline{\issue_2}$ wins. It is easy to see that this establishes the equivalence of \Knapsack{} and {\sc Weighted-}$(\SM, A, C)$-\destBrib.

For the combination of~$\Cdist$ and~$\DV$ we need to introduce a new issue~$\Issue_3$. For each voter with a cost vector entry of $2\budget+1$, we add the independent cp-statement~$\overline{\issue_3} > \issue_3$; for all the other voters (created for an object) we add the independent cp-statement~$\issue_3>\overline{\issue_3}$ and change the preference for $\Issue_1$ to $\overline{\issue_1}>\issue_1$. The needed order~$\mathcal{O}$ over the issues is extended to $\Issue_1 > \Issue_2 > \Issue_3$. The candidate~$\overline{\issue_1}\,\overline{\issue_2}\,\overline{\issue_3}$ serves as the disliked candidate~$h$, and, since the broad majority of the voters votes for~$\overline{\issue_3}$ over~$\issue_3$, he wins the unbribed election. To prevent this, the briber can, again, only bribe an appropriate subsets of voters---all created for an object---to flip the cp-statement~$\overline{\issue_1}:\overline{\issue_2}>\issue_2$. This will rise the candidate~$\overline{\issue_1}\issue_2\issue_3$ from the third to the top position for those voters, resulting in a total cost of~$2\cdot w_i$ to bribe such a voter~$v_i$. Because of this factor of~$2$ we need to double the budget, too. If there is a solution to the corresponding instance of \Knapsack{}, there exists a bribery letting candidate~$\overline{\issue_1}\issue_2\overline{\issue_3}$ win the election. 

By this modification, we ensure, apart from introducing costs to the needed bribery, that, again, no voter initially voting for~$h$ can be bribed. Therefore this reduction holds for the negative case as well. 
\end{proof}

In contrast, the combination of $\OV$ with $\Cdist$ can be solved in polynomial time by a greedy-algorithm because all possible bribery actions are for free.

\begin{theorem}\label{thm:neg:weig:OVDVCdist:P}
\weigh-$(\OV, \DV, \Cdist)$-$\langle\negative\rangle$-\destBrib is solvable in polynomial time.
\end{theorem}

\begin{proof}
Consider an instance of \weigh-$(\OV, \DV, \Cdist)$-\negative-\destBrib. We distinguish two cases. (1) If the disliked candidate $h$ does not get a single veto, he cannot be prevented from winning at all. (2) Otherwise, he gets at least one veto. This case can be subdivided into two sub-cases: (2.1) If there is a candidate without any vetos, there is nothing to do, because $h$ is not a winner. (2.2) But if every candidate gets at least one veto, we might be able to prevent~$h$ from winning. Note that in this case there cannot be more candidates than voters, so $n\geq 2^m$ must hold. Therefore we can try to make each of the candidates a winner. This can easily be done by bribing as many voters to cast their veto against~$h$ as possible. After that we try for each candidate (except~$h$) to take as many vetos from this candidate as possible. Since we are only allowed to bribe dependent issues, this will involve no costs with the cost scheme $\Cdist$, because each voter can only be bribed such that the top candidate remains the same. This is due to the fact that one cp-statement of each dependent issue is needed to evaluate the top candidate and another cp-statement is used to determine the least favorite candidate (who will get the veto). So just flipping the latter will result in a different candidate getting the veto, but will preserve the top candidate. Note that it is possible to bribe a voter in more than just one dependent issue, but in our case it is not necessary. 

This procedure can be used to solve the non-negative case, too. We only need to bribe as many voters as possible to cast their veto against~$h$. This again involves no costs, and after this preprocessing we can apply the same procedure as above.
\end{proof}

\section{Special Cases: Setting each entry in the cost vector to the same value.}
\label{sec:CostVectorSameValue}

First we argue that if the cost vector does not distinguish different voters and sets each cost to $a$, we can assume each entry to be $1$. This is obviously true as we can just divide the budget and each entry of the cost vector by $a$ which forces us to round all values. It is easy to see that this does not change the instance as we can not spend the budget we lose through rounding. Even a single bribery is more expensive than the budget we lose.

\begin{theorem}\label{thm:opivivdvflip_withoutcostvector}
$(OP, A, \Cflip)$-\textsc{destructive-Bribery} is in \Pclass with bribery action $A\in\{IV, IV+DV\}$, if the cost vector assigns the same value to each voter.
\end{theorem}

\begin{proof}
For the bribery actions $IV$ and $IV+DV$, each voter can be bribed to vote for at least one different candidate with a cost of~$1$. Note that it is never helping the briber to bribe a voter for more than cost~$1$, because for this cost he can always decrease the score of $h$ by one and/or increase the score of the designated winner. 
Moreover, existence of a successful bribery is equivalent to existence of a successful bribery which has a cost of 1 per voter and where voters are only asked not to vote for~$h$.

There are at most $n\cdot m$ different candidates the briber can bribe voters to vote for. 
Since they are easy to identify one can calculate for every such candidate the costs to let this candidate win by just bribing voters to vote to him or/and to some arbitrary other candidate instead of $h$. The cheapest such bribery is the solution if the budget is sufficient, otherwise there exists none.
\end{proof}

\begin{corollary}\label{cor:opivdvcdist}
$(OP, IV+DV, \Cdist)$-\textsc{destructive-Bribery} is in \Pclass if the cost vector assigns the same value to each voter.
\end{corollary}

\begin{proof}
The same argument used in Theorem~\ref{thm:opivivdvflip_withoutcostvector} holds here. There is always exactly one candidate the briber can bribe a voter to vote for instead of $h$ with costs $1$ (by bribing the least important issue). 
\end{proof}

\begin{corollary}\label{cor:opivdvclevel}
$(OP, IV+DV, \Clevel)$-\textsc{destructive-Bribery} is in \Pclass if the cost vector assigns the same value to each voter.
\end{corollary}

\begin{proof}
Again, the same argument used in Theorem~\ref{thm:opivivdvflip_withoutcostvector} holds. There is always at least one candidate the briber can bribe a voter to vote for instead of $h$ with costs $1$ (this will be a depending issue as long as the voter does not only have independent issues). 
\end{proof}

\begin{theorem}\label{thm:okivdvcflip}
$(\OKeff, IV+DV, \Cflip)$-\textsc{destructive-Bribery} is in \Pclass{} for $\mathcal{O}$-legal profiles if the cost vector assigns the same value to each voter.
\end{theorem}

\begin{proof}
Let $\issue_1\issue_2\dots \issue_{m-1} \issue_{m}$ be the disliked candidate~$h$ and~$c^*$ denote the candidate $\issue_1\issue_2\dots \issue_{m-1},\overline{\issue_m}$. In every instance the score of~$h$ is greater than the score of~$c^*$, 
therefore there has to be at least one voter having~$h$ but not~$c^*$ among his first~$k$ preferred candidates. This can happen if and only if the rank of~$h$ in this voter's induced preference order over all candidates is~$k$, and the rank of $c^*$ is $k+1$.
For each of these voters the briber needs to bribe just one flip on the last issue to reduce the score of~$h$ and, simultaneously, increase the score of~$c^*$ by one. Note that there is no candidate~$c_i\neq c^*$ with $s(c_i)\leq\lfloor\frac{s(h)}{2}\rfloor$, which can be made a winner with less flips than $c^*$. Therefore it is sufficient to calculate the costs to make $c^*$ or any candidate~$c_j$, with initially $s(c_j) > \lfloor\frac{s(h)}{2}\rfloor$, win the election. The cheapest such bribery is the solution, if the budget is sufficient.
Since there are at most~$n$ candidates to check, each by iterating over~$n$ voters with $m$ issues, this results in a running-time in~$\bigO{n^2\cdot m}$.
\end{proof}

\begin{theorem}\label{thm:okivcflip}
$(\OKeff, IV, \Cflip)$-\textsc{destructive-Bribery} is in \Pclass{} if the cost vector assigns the same value to each voter.
\end{theorem}

\begin{proof}
If $k$ is polynomial in $n$ and $m$, each voter initially voting for $h$ can be bribed to not vote for $h$ by just flipping one issue. This can be the most important one, but some other issue as well. However, the most important one is always an independent issue and therefore bribable. So it is sufficient to bribe just voters initially voting for $h$ with exactly one flip. There are at most $m\cdot k$ candidates which can get an additional vote from one voter, so at most $m\cdot k \cdot n$ candidates can beat $h$ after the bribery (plus possibly the candidate with the second most number of votes before the bribery).

Note that it is not necessary to try each of those candidates to win. Just track for each such candidate how many voters voting for $h$ can be bribed to vote for him with only one flip. One can do so by iterating over the voter voting for $h$ and increasing a counter for each candidate the voter can be bribed to vote for by only one flip.  Let $a_i$ be the number of such voters who can be bribed to vote for the candidate $c_i$. If there is a candidate with $s(h) < s(c_i)+a_i+\budget$, this candidate can be made a winner by bribing $a_i$ voters to vote for $c_i$ but not for $h$ and additional $\budget-a_i$ voters  not to vote for $h$. If there is no such voter, it is a no-instance. 
\end{proof}

\begin{theorem}\label{thm:okivdvcdist}
$(\OKeff, IV+DV, \Cdist)$-\textsc{destructive-Bribery} is in \Pclass{} for $\mathcal{O}$-legal profiles if the cost vector assigns the same value to each voter.
\end{theorem}

\begin{proof}
We prove Theorem~\ref{thm:okivdvcdist} with a similar idea as in the proof of Theorem~\ref{thm:okivdvcflip}. First we preprocess the instance: whenever a voter voting for~$h$ is given who can be bribed no to vote for~$h$ without costs, we bribe him in this way, as this is never wrong. Note that this can be done efficiently, because there are at most $l=\lceil\log k\rceil$ many issues that may be considered for this, so even a brute-force approach would be doable in polynomial time. 

Let the candidate $\issue_1\issue_2\dots\issue_{m-1}\overline{\issue_m}$ be denoted by~$c^*$, whereas the disliked candidate~$h$ is $\issue_1\issue_2\dots\issue_{m-1}\issue_{m}$. In every instance 
the score of $h$ is greater than the score of $c^*$, 
therefore there has to be at least one voter ranking~$h$ but not~$c^*$ among his first~$k$ preferred candidates. This can happen if and only if~$h$ is on rank~$k$ and~$c^*$ on rank~$k+1$. For each of these voters the briber needs to bribe just one flip on the last issue to reduce the score of~$h$ and, simultaneously, increase the score of~$c^*$ by one. This flip will always result in a change of the top candidate, which would imply a cost of $1$, because the top two candidates will swap their position.

Note that, depending on~$k$, one might change the score of certain candidates without spending budget, because the top candidate stays the same. This can only happen due to flips in the last $l$ issues, and there only with dependent ones. Therefore there are at most $n\cdot 2^l$ such candidates, whose scores can be raised without both paying anything and raising the score of $h$. Let $s_\text{max}(c_i)$ denote the maximum score the candidate $c_i$ can reach by such a bribery.

There is no candidate~$c_i\neq c^*$ with $s_\text{max}(c_i)\leq\lfloor\frac{s(h)}{2}\rfloor$ which can be made a winner with less costs than $c^*$. Therefore it is sufficient to calculate the costs to make $c^*$ or any candidate~$c_j$, with initially $s_\text{max}(c_j) > \lfloor\frac{s(h)}{2}\rfloor$, win the election. The cheapest such bribery is the solution, if the budget is sufficient.
Since there are less than $n\cdot (2^l +1)$ candidates $c_i$ with $s_\text{max}(c_j)>0$, there are at most~$n\cdot (2^l +1)$ candidates to check, each by iterating over~$n$ voters with $m$ issues, resulting in a running-time in~$\bigO{n^2\cdot (2^l +1)\cdot m}$.
\end{proof}

\emph{Comment.} $\mathcal{O}$-legality of the profile is needed for Lemma~\ref{thm:okivdvcdist}, because otherwise there is not necessarily a unique~$c^*$ which is next to~$h$ in the order of candidates for each voter.

\begin{theorem}\label{thm:okivivdvcflipkunbounded}
$(\OKeff, A, \Cflip)$-\textsc{destructive-Bribery} is in \Pclass{} with bribery action $A\in \{IV, IV+DV\}$ for $\mathcal{O}$-legal profiles and $k$ being a power of~$2$, if the cost vector assigns the same value to each voter.
\end{theorem}

\begin{proof}
With $k=2^j$ it is never helping the briber to bribe one of the last $j$ issues. Therefore a successful bribery makes a candidate win that differs from $h$ in at least one of the first $m-j$ issues. Such a candidate can be found in polynomial running-time with the algorithm given in the proof of Theorem~\ref{thm:opivivdvflip_withoutcostvector}, from the voting protocol $OP$.
\end{proof}

\begin{theorem}\label{thm:okivivdvcdistkunbounded}
$(\OKeff, IV+DV, \Cdist)$-\textsc{destructive-Bribery} is in \Pclass{}, when $k$ is a power of~$2$, if the cost vector assigns the same value to each voter.
\end{theorem}

\begin{proof}
With $k=2^j$ it is never helping the briber to bribe one of the last $j$ issues. So the cheapest reasonable issue to bribe is issue $\Issue_{m-j}$. With $\Cdist$ the next best issue would be issue $\Issue_{m-j-1}$, but this will cost twice as much. Therefore the best strategy is to bribe voters initially voting for $h$ only in issue $\Issue_{m-j}$ until~$h$ loses against some other candidate or until the budget is not sufficient anymore. Note that the voters cannot be chosen at random, but it can be easily checked which one should be started with. 
\end{proof}

\section{Conclusion}\label{sec:conclusion}
We extensively studied destructive bribery for the weighted, unweighted, negative and non-negative variations on all cost-, bribery- and evaluation-schemes introduced by Mattei~\etal~\cite{mattei2013bribery}. Table \ref{table:results} summarizes our results. The main differences can be observed between the weighted and unweighted cases, while the negative and non-negative cases are very similar---we remark that these cases may behave differently, however, if the cost vector assigns the same value to each of the voters. The cost vector is also the tool to mimic the restriction of the negative setting in the non-negative case: with its help, one can make a voter unbribable (one cannot use it to affect the bribery actions though).

It is also interesting to observe that---for an arbitrary cost vector---all combinations in the weighted case for the {\it constructive} bribery problem turned out to be \NP-complete~\cite{dorn2014hardness}, whereas in the destructive setting, we have identified two tractable cases. They occur due to the strange side effect of the combination of $\DV$ and $\Cdist$ that sets all reasonable bribery free of charge. 

The most interesting observation might be that in the unweighted setting, only the combination of $\SM$ and $\Cequal$ was shown to be \NP-complete for constructive bribery~\cite{mattei2013bribery,dorn2014hardness}, while almost half of the combinations for destructive bribery turned out to be computationally hard -  this behavior is rather unusual for voting problems and is due to the combinatorial structure of the set of candidates. If the number of candidates is part of the input (which is the case for many of the common settings for voting problems), the constructive case of a voting problem can be directly used to solve the destructive counterpart: If it is known how to make a designated candidate the only winner of the election, one can simply run this procedure for all candidates and find out which of them is the cheapest solution. 
In the setting of combinatorial domains where the number of candidates is exponential in the size of the input, one cannot simply check which of the exponentially many candidates should be chosen to be made the winner. It might turn out, nevertheless, that the destructive version is not computationally harder, but in our case, we have seen that precluding an alternative might be more difficult than pushing it through.

\bibliographystyle{plain}
\bibliography{literatur}
\end{document}